%% file: threshold-rank.tex
\documentclass[11pt]{article}
\usepackage{tr-macros}

\renewcommand{\epsilon}{\eps}

\allowdisplaybreaks

\begin{document}
\title{Coloring 3-Colorable Graphs with Low Threshold Rank}

\author{Jun-Ting Hsieh\thanks{Carnegie Mellon University. \texttt{juntingh@cs.cmu.edu}. Supported by NSF CAREER Award \#2047933.}}


\date{\today}

\maketitle

\begin{abstract}
    We present a new algorithm for finding large independent sets in $3$-colorable graphs with small $1$-sided threshold rank.
    Specifically, given an $n$-vertex $3$-colorable graph whose uniform random walk matrix has at most $r$ eigenvalues larger than $\varepsilon$, our algorithm finds a proper $3$-coloring on at least $(\frac{1}{2}-O(\varepsilon))n$ vertices in time $n^{O(r/\varepsilon^2)}$.
    This extends and improves upon the result of Bafna, Hsieh, and Kothari \cite{BHK25} on $1$-sided expanders.
    Furthermore, an independent work by Buhai, Hua, Steurer, and Vári-Kakas~\cite{BHSV25} shows that it is UG-hard to properly $3$-color more than $(\frac{1}{2}+\varepsilon)n$ vertices, thus establishing the tightness of our result.

    Our proof is short and simple, relying on the observation that for any distribution over proper $3$-colorings, the correlation across an edge must be large if the marginals of the endpoints are not concentrated on any single color.
    Notably, this property fails for $4$-colorings, which is consistent with the hardness result of \cite{BHK25} for $4$-colorable $1$-sided expanders.
\end{abstract}




\input{tr-intro}
\input{tr-prelims}

\input{tr-proof}

\input{tr-is}

\section*{Acknowledgements}

The author would like to thank Mitali Bafna and Pravesh K.\ Kothari for helpful discussions and their support in posting this paper.
The author would also like to thank Rares-Darius Buhai, Yiding Hua, David Steurer, and Andor Vári-Kakas for the communication and exchange of results regarding their work~\cite{BHSV25}, and for pointing out an error in an earlier version.

\bibliographystyle{alpha}
\bibliography{threshold-rank}

\end{document}

%% file: tr-intro.tex
\section{Introduction}

Coloring $3$-colorable graphs is a classical and challenging problem in graph algorithms.
Starting with the seminal work of Wigderson~\cite{Wigderson83}, a long line of research \cite{Blum94,BlumK97,KMS98,ACC06,Chlamtac09,KawarabayashiT17,KTY24} has focused on finding polynomial-time algorithms that color $3$-colorable graphs using as few colors as possible,
with the current best achieving an $\wt{O}(n^{0.197})$-coloring~\cite{KTY24}.
The difficulty in improving the above algorithms may be inherent due to strong hardness results.
Bansal and Khot \cite{BansalK09} proved that even if the graph is almost $2$-colorable (that is, becomes $2$-colorable after removing any constant fraction of the vertices), it is NP-hard to find an independent set of linear size, assuming the Unique Games Conjecture.
This rules out any algorithm that colors such graphs with any constant number of colors.
For exactly $3$-colorable graphs, similar hardness results are known under variants of the Unique Games Conjecture (with perfect completeness) \cite{DinurS05,DinurMR06,DinurKPS10,KhotS12,GuruswamiS20}.

In light of the hardness results for worst-case instances, a line of work has focused on algorithms for $3$-colorable graphs that satisfy additional structural guarantees.
One direction explores planted random models for coloring~\cite{BlumS95,AlonK97}, while another studies colorable graphs with expansion properties~\cite{AroraG11,DavidF16,KumarLT18}.
Notably, David and Feige \cite{DavidF16} gave polynomial-time algorithms for finding linear-sized independent sets in planted $k$-colorable graphs that are \emph{$2$-sided spectral expanders} --- where all non-trivial eigenvalues $1 \geq \lambda_2 \geq \lambda_3 \geq \cdots \geq \lambda_n \geq -1$ of the uniform random walk matrix are bounded \emph{in absolute value}.
In fact, all these works~\cite{AroraG11,DavidF16,KumarLT18} crucially rely on bounds on the negative eigenvalues of the graph.

Recently, Bafna, Hsieh and Kothari~\cite{BHK25} studied $3$-colorable \emph{$1$-sided expanders}, without any assumption on the negative spectrum (i.e., only assuming lower bounds on the spectral gap of the Laplacian).
Surprisingly, they showed that the complexity of the problem is drastically different between $3$- and $\geq 4$-colorable 1-sided expanders.
Specifically, they gave an algorithm that finds an independent set of size $\geq cn$ in an almost $3$-colorable graph with $\lambda_2 \leq c$, for some small universal constant $c$.
On the other hand, they proved, assuming the Unique Games Conjecture, that it is NP-hard to find an $\Omega(n)$-sized independent set in an almost $4$-colorable graph, even when the graph has nearly perfect $1$-sided expansion, i.e., $\lambda_2 \leq o(1)$.
This is in sharp contrast to previous works \cite{AroraG11,DavidF16,KumarLT18} whose techniques extend to $k$-colorable $2$-sided expanders for all constant $k$.

A natural next step is to consider \emph{low threshold rank graphs} with upper bounds on $\lambda_r$ for some $r > 2$, which generalize $1$-sided expanders (the $r=2$ case).
Many interesting graph families have small threshold rank, including small-set expanders and hypercontractive graphs \cite{ABS15,Steurer10}.
For Unique Games, algorithms developed for expanders and then for low threshold rank graphs \cite{Trevisan08,AroraKKSTV08,MakarychevM11,Kolla11} directly led to the breakthrough subexponential-time algorithm for general instances \cite{ABS15}.
In a similar spirit, designing coloring algorithms for low threshold rank graphs may offer a potential stepping stone toward subexponential-time algorithms for coloring worst-case $3$-colorable graphs with $n^{o(1)}$ or even $\polylog(n)$ colors.


\subsection{Our results}

We extend the result of \cite{BHK25} to graphs with low threshold rank, again without any assumption on the negative spectrum.
Moreover, we obtain improved guarantees on the size of the coloring, and our proof is much simpler compared to \cite{BHK25}.

\parhead{3-coloring.}
A proper $3$-coloring on $m$ vertices refers to a partial coloring in which $m$ vertices are colored, and every edge with both endpoints colored must have different colors.
Equivalently, it is a collection of $3$ disjoint independent sets whose union has size $m$.
We say that a graph is $\delta$-almost $3$-colorable if there is a proper $3$-coloring on $1-\delta$ fraction of the vertices.

\begin{mtheorem} \label{thm:main}
    Let $\eps \in (0,1)$, $\delta \in [0,1]$ and $r, n \in \N$.
    There is an algorithm that, given an $n$-vertex regular graph that is $\delta$-almost $3$-colorable with at most $r$ eigenvalues larger than $\eps/100$, outputs a proper $3$-coloring on at least $(\frac{1}{2} - \eps - O(\delta))n$ vertices in time $n^{O(r/\eps^2)}$.
\end{mtheorem}

For comparison, the result of \cite{BHK25} finds only one independent set of size $10^{-4}n$ while assuming $\lambda_2 \leq 10^{-4}$, and they showed that the same guarantee is UG-hard for almost $4$-colorable graphs even if $\lambda_2 \leq o(1)$.
The algorithm of \cite{BHK25} relies on a novel combinatorial clustering property of proper $3$-colorings on $3$-colorable $1$-sided expanders (along with involved analyses to formalize it in the Sum-of-Squares proof system).
Unfortunately, their techniques do not seem to generalize to low threshold rank graphs in a natural way.

In a concurrent and independent work, Buhai, Hua, Steurer, and Vári-Kakas~\cite{BHSV25} showed that for every $\eps>0$, given an $\eps$-almost $3$-colorable graph with $\lambda_2\leq \eps$, it is UG-hard to find a proper $(1/\eps)$-coloring for at least a $(\frac{1}{2}+\eps)$ fraction of vertices.
Therefore, \Cref{thm:main} is essentially \emph{tight}, matching their hardness result at an intriguing threshold of coloring $1/2$ of the vertices.

\cite{BHSV25} also showed several remarkable results that improve on \cite{BHK25}.
Notably, given a $3$-colorable $1$-sided expander with an extra assumption that all color classes have size at most $(1/2-\eps)n$, their algorithm finds a proper $3$-coloring for $1-\eps$ fraction of vertices, where $\eps \to 0$ as $\lambda_2 \to 0$.
On the other hand, they proved that a similar guarantee is not possible under the weaker guarantee of low threshold rank:
even if all color classes have size at most $n/3$ and only assuming $\lambda_3 \leq \eps$, it is UG-hard to find a proper ($1/\eps)$-coloring for at least a $0.9$ fraction of vertices.
This shows that it is unlikely to improve \Cref{thm:main} much even under the guarantee that the color classes are balanced.

\parhead{Large independent set.}
Similar to \cite{BHK25}, our main ideas extend to finding large independent sets in low threshold rank graphs that contain independent sets of size $(\frac{1}{2}-\delta)n$
\footnote{When the input graph has an independent set of size $(\frac{1}{2}+\delta)n$,
the factor-2 approximation algorithm for vertex cover finds an independent set of size $2\delta n$.}
(in fact, this is the easier case).

\begin{mtheorem} \label{thm:main-is}
    Let $\eps,\delta \in (0,1)$ and $r, n \in \N$.
    There is an algorithm that, given an $n$-vertex regular graph that contains an independent set of size $(\frac{1}{2}-\delta)n$ and has at most $r$ eigenvalues larger than $\eps^5/100$,
    outputs an independent set of size at least $(\frac{1}{2} - 2\delta - \eps)n$ in time $n^{O(r/\eps^{10})}$.
\end{mtheorem}

For comparison, the result of \cite{BHK25} finds an independent set of size $10^{-3}n$ while assuming $\lambda_2 \leq 1-40\delta$.
We are able to find an independent set of size almost $n/2$, but this requires $\lambda_r$ to be close to $0$.

\subsection{Brief overview}

An interesting challenge for the $3$-coloring problem is to design a strategy that works for $3$-colorable $1$-sided expanders but fails for $4$-colorable ones, due to the hardness result for $4$-colorable expanders \cite{BHK25}.
For \cite{BHK25}, they discovered the combinatorial clustering property of $3$-colorings that does not hold for $4$-colorings.
In this work, we identify a simple correlation property, which we describe next.

Suppose we have a distribution $\mu$ over proper $3$-colorings of the graph (with domain $[3]^n$), where for every edge $(u,v)$ we have $\Pr_{\mu}[X_u = X_v] = 0$.
It is easy to verify that the set $S_{\sigma} = \{u\in[n]: \Pr_{\mu}[X_u = \sigma] > 1/2\}$ forms an independent set for each $\sigma\in [3]$ (see, e.g., \cite[Fact 3.8]{BHK25}).
Let $T = [n] \setminus (S_1 \cup S_2 \cup S_3)$ be the ``bad'' set (where $X_u$ is not concentrated on any color), and we would like to upper bound $|T|$.

Our main observation is that for any edge $(u,v)$ with both $u,v \in T$, the random variables $X_u$ and $X_v$ must have a large ``correlation''.
See \Cref{lem:mutual-info-lb} for a precise statement.
We emphasize that an analogous result for $4$-coloring does not hold; see \Cref{rem:4-coloring-fails}.

If $|T| > (\frac{1}{2}+\eps)n$, then there must be at least an $\Omega(\eps)$ fraction of edges with both endpoints in $T$.
This means that the ``local correlation'' of $\mu$ --- the expected correlation over the edges --- is at least $\Omega(\eps)$.
Then, it is known \cite{BRS11} that large local correlation implies large global correlation in low threshold rank graphs (\Cref{lem:local-correlation}).
On the other hand, a standard technique known as global correlation rounding~\cite{BRS11,GuruswamiS11,RT12} shows that iteratively conditioning on the vertices can reduce the global correlation to be arbitrarily small, which contradicts large global correlation.

Thus, our algorithm is simple:
solve the degree-$O(r/\eps^2)$ Sum-of-Squares relaxation and obtain a pseudo-distribution $\mu'$ (see \Cref{sec:sos} for background).
Next, condition on $\mu'$ to obtain a pseudo-distribution $\mu$ with a small enough global correlation.
Then, color the vertices in $S_1, S_2, S_3$ accordingly.
The above argument shows that $S_1 \cup S_2 \cup S_3$ must have size at least $(\frac{1}{2}-\eps) n$.

%% file: tr-prelims.tex
\section{Preliminaries}

We will need the following standard inequality.

\begin{fact}[Pinsker's inequality] \label{fact:pinskers}
    $\|P-Q\|_1 \leq \sqrt{2\KL(P\|Q)}$.
    Consequently, for joint distributions $X,Y$, $\|(X,Y) - X \otimes Y\|_1^2 \leq 2 I(X;Y)$.
\end{fact}

\subsection{Background on Sum-of-Squares}
\label{sec:sos}

We refer the reader to the monograph~\cite{FKP19} and the lecture notes~\cite{BS16} for a detailed exposition of the Sum-of-Squares (SoS) method. 

\parhead{Pseudo-distributions.}
Pseudo-distributions are generalizations of probability distributions.
Formally, a pseudo-distribution on $\R^n$ is a finitely supported \emph{signed} measure $\mu :\R^n \rightarrow \R$ such that $\sum_{x} \mu(x) = 1$. The associated \emph{pseudo-expectation} is a linear operator $\pE_\mu$ that assigns to every polynomial $f:\R^n \rightarrow \R$ the value $\pE_\mu f = \sum_{x} \mu(x) f(x)$, which we call the pseudo-expectation of $f$. We say that a pseudo-distribution $\mu$ on $\R^n$ has \emph{degree} $d$ if $\pE_\mu[f^2] \geq 0$ for every polynomial $f$ on $\R^n$ of degree $\leq d/2$.

A degree-$d$ pseudo-distribution $\mu$ is said to satisfy a constraint $\{q(x) \geq 0\}$ for any polynomial $q$ of degree $\leq d$ if for every polynomial $p$ such that $\deg(p^2) \leq d-\deg(q)$, $\pE_\mu[p^2 q] \geq 0$.
We say that $\mu$ $\tau$-approximately satisfies a constraint $\{q \geq 0\}$ if for any sum-of-squares polynomial $p$, $\pE_\mu[pq] \geq - \tau \Norm{p}_2$ where $\Norm{p}_2$ is the $\ell_2$ norm of the coefficient vector of $p$. 

We rely on the following connection that forms the basis of the sum-of-squares algorithm.

\begin{fact}[Sum-of-Squares algorithm, \cite{Par00,Las01}] \label{fact:sos-algorithm}
Given a system of degree $\leq d$ polynomial constraints $\{q_i \geq 0\}$ in $n$ variables and the promise that there is a degree-$d$ pseudo-distribution satisfying $\{q_i \geq 0\}$ as constraints, there is a $n^{O(d)} \polylog( 1/\tau)$ time algorithm to find a pseudo-distribution of degree $d$ on $\R^n$ that $\tau$-approximately satisfies the constraints $\{q_i \geq 0\}$.
\end{fact}

We may set the error parameter $\tau$ to be $2^{-\poly(n)}$, which is small enough for our analysis.
We will thus omit discussions of numerical precision in this paper.

\parhead{SoS relaxation of $3$-coloring.}
We consider the domain $X \in \Sigma^n$ for some finite alphabet $\Sigma$.
The variables are $\{Y_{u,\sigma}\}_{u\in[n], \sigma\in \Sigma}$, and we have the following constraints:
for all $u\in [n]$ and $\sigma\neq \sigma' \in \Sigma$,
(1) $Y_{u,\sigma} = Y_{u,\sigma}^2$,
(2) $Y_{u,\sigma} Y_{u,\sigma'}=0$,
and (3) $\sum_{\sigma\in \Sigma} Y_{u,\sigma} = 1$.
Intuitively, $Y_{u,\sigma}$ is the $\{0,1\}$ indicator that vertex $u$ is assigned $\sigma$.
If $\Sigma$ is the set of colors, then the coloring constraints include $Y_{u,\sigma} Y_{v,\sigma} = 0$ for all $(u,v)\in E(G)$ and $\sigma\in\Sigma$.

For ease of notation, we will use $\1(X_u = \sigma)$ to refer to $Y_{u,\sigma}$.
Moreover, given a pseudo-distribution $\mu$ that satisfies the above constraints, we will write $\pPr_{\mu}[X_u = \sigma] = \pE_{\mu}[Y_{u,\sigma}]$ and similarly
$\pPr_{\mu}[X_u = \sigma, X_v = \sigma'] = \pE_{\mu}[Y_{u,\sigma} Y_{v,\sigma'}]$.

Given a pseudo-distribution $\mu$ of degree $O(k|\Sigma|)$, the marginal of $\mu$ over any set of $k$ vertices corresponds to a true distribution on $\Sigma^k$.
Thus, standard statistical quantities involving random variables $X_u, X_v$ --- such as the mutual information $I_{\mu}(X_u; X_v)$ --- are well defined.

\subsection{Global correlation rounding}

An standard technique we need is reducing the average correlation of random variables through conditioning, which was introduced in \cite{BRS11} (termed global correlation reduction).
It is also applicable to pseudo-distributions of sufficiently large degree.
We will use the following version from~\cite{RT12}.

\begin{lemma}[\cite{RT12}]
\label{lem:ragh-tan}
Let $X_1,\ldots, X_n$ be a set of random variables each taking values in $\{1,\ldots, q\}$.
Then, for any $\ell \in \N$, there exists $k \leq \ell$ such that: 
\begin{align*}
    \E_{i_1,\ldots,i_k \sim [n]}\E_{i,j \sim [n]}[I(X_i;X_j \mid X_{i_1}, \ldots,X_{i_k})] \leq \frac{\log q}{\ell - 1} \mper
\end{align*}
\end{lemma}
Note that the above lemma holds as long as there is a collection of local distributions over $(X_1,\ldots, X_n)$ that are valid probability distributions over all subsets of $\ell+2$ variables and are consistent with each other.
In particular, this is satisfied by a pseudo-distribution $\mu$ of degree $\geq \ell+2$ over the variables $(X_1,\ldots, X_n)$.

\subsection{Local to global in low threshold rank graphs}
The following lemma is often called the ``local-to-global'' inequality (see \cite[Lemma 6.1]{BRS11}).
We give a proof for completeness.

\begin{lemma}[Local to global] \label{lem:local-to-global}
    Let $M \in \R^{n\times n}$ be positive semidefinite such that $\tr(M) \leq n$.
    Let $G$ be a regular graph on $n$ vertices whose random walk matrix has at most $r$ eigenvalues larger than $\lambda > 0$.
    Then, $\E_{(i,j) \sim G}[M_{ij}] \leq (1-\lambda)\sqrt{r\cdot\E_{i,j\sim [n]}[M_{ij}^2]} + \lambda$.
\end{lemma}
\begin{proof}
    Let $A$ be the random walk matrix with eigenvalues $1 = \lambda_1 \geq \lambda_2 \geq \cdots \geq \lambda_n$,
    and let $A = \sum_{i=1}^n \lambda_i v_i v_i^\top$ where $v_i$'s are the unit-norm eigenvectors.
    By assumption, we have $\lambda_{r+1} \leq \lambda$.
    Thus, $A \preceq \sum_{i=1}^r v_i v_i^\top + \lambda \sum_{i=r+1}^n v_i v_i^\top = (1-\lambda) \sum_{i=1}^r v_i v_i^\top + \lambda \Id_n$, where we use the fact that $\sum_{i=1}^n v_i v_i^\top = \Id_n$.
    Then,
    \begin{align*}
        \E_{(i,j)\sim G}[M_{ij}]
        &= \frac{1}{n} \angles*{M, A}
        \leq \frac{1-\lambda}{n} \angles*{M,\ \sum_{i=1}^r v_i v_i^\top} + \frac{\lambda}{n} \tr(M)
        \leq \frac{1-\lambda}{n} \|M\|_F \sqrt{r} + \frac{\lambda}{n} \tr(M) \mcom
    \end{align*}
    where the last inequality uses the Cauchy-Schwarz inequality and $\|\sum_{i=1}^r v_i v_i^\top\|_F = \sqrt{r}$.
    Then, since $\|M\|_F = \sqrt{\sum_{i,j\in[n]} M_{ij}^2}$ and $\tr(M) \leq n$, this completes the proof.
\end{proof}

\Cref{lem:local-to-global} will be used in the next lemma that bounds the ``local correlation'' of a pseudo-distribution with small ``global correlation''.
The proof is similar to that of \cite[Lemma 5.4]{BRS11}

\begin{lemma}
\label{lem:local-correlation}
    Let $G$ be a regular graph on $n$ vertices whose random walk matrix has at most $r$ eigenvalues larger than $\lambda > 0$.
    Let $\Sigma$ be a finite alphabet, and let $\mu$ be a pseudo-distribution over $\Sigma^n$ of degree $O(|\Sigma|)$.
    Suppose the global correlation $\E_{u,v\sim [n]}[I_{\mu}(X_u; X_v)] \leq \delta$.
    Then,
    \begin{align*}
        \E_{(u,v)\sim G} \sum_{\sigma\in \Sigma} 
        \parens*{ \pPr_{\mu}[X_u = X_v = \sigma] - \pPr_{\mu}[X_u = \sigma] \pPr_{\mu}[X_v = \sigma] }^2
        \leq \sqrt{2r\delta} + \lambda \mper
    \end{align*}
\end{lemma}

\begin{proof}
    For $\sigma\in \Sigma$, let $M_{\sigma} \in \R^{n\times n}$ be the matrix with entries $(M_{\sigma})_{uv} = \pPr_{\mu}[X_u = X_v = \sigma] - \pPr_{\mu}[X_u = \sigma] \pPr_{\mu}[X_v = \sigma]$.
    First note that $M_{\sigma} \succeq 0$.
    To see this, consider the vector $z$ such that $z_u = \1(X_u = \sigma) - \pPr_{\mu}[X_u = \sigma]$.
    Then, $(M_{\sigma})_{uv} = \pE_{\mu}[z_u z_v]$ and hence $M_{\sigma} = \pE_{\mu}[zz^\top] \succeq 0$.

    Now, let $\wt{M} \in \R^{n\times n}$ where $\wt{M}_{uv} = \sum_{\sigma\in \Sigma}(M_{\sigma})_{uv}^2$.
    The lemma statement can be simplified to
    \begin{align*}
        \E_{(u,v)\sim G}\bracks*{\wt{M}_{uv}} \leq \sqrt{2r\delta} + \lambda \mper
    \end{align*}
    Since $M_{\sigma} \succeq 0$, squaring each of its entries also gives a positive semidefinite matrix.
    Thus, $\wt{M} \succeq 0$.
    Moreover, the diagonal entries $\wt{M}_{uu} = \sum_{\sigma} (p_{u,\sigma}-p_{u,\sigma}^2)^2 \leq 1$, where $p_{u,\sigma} = \pPr_{\mu}[X_u=\sigma]$.
    This implies that the off-diagonal entries satisfy $0 \leq \wt{M}_{uv} \leq 1$ as well.
    
    Next, for any $u,v \in [n]$,
    \begin{align*}
        \norm*{(X_u, X_v) - X_u \otimes X_v}_1^2
        &= \parens*{ \sum_{\sigma, \sigma'\in \Sigma} \abs*{ \pPr_{\mu}[X_u = \sigma, X_v = \sigma'] - \pPr_{\mu}[X_u = \sigma] \pPr_{\mu}[X_v = \sigma'] } }^2 \\
        &\geq \sum_{\sigma\in \Sigma}  \parens*{ \pPr_{\mu}[X_u = X_v = \sigma] - \pPr_{\mu}[X_u = \sigma] \pPr_{\mu}[X_v = \sigma] }^2 \\
        &= \wt{M}_{uv} \mper
    \end{align*}
    By Pinsker's inequality (\Cref{fact:pinskers}), $\norm*{(X_u, X_v) - X_u \otimes X_v}_1^2 \leq 2\cdot I(X_u ; X_v)$.
    Thus, global correlation being at most $\delta$ implies that
    \begin{align*}
        \E_{u,v \sim [n]} \bracks*{\wt{M}_{uv}^2} 
        \leq \E_{u,v \sim [n]} \bracks*{\wt{M}_{uv}} \leq 2\delta \mper
    \end{align*}
    Then, applying \Cref{lem:local-to-global} to $\wt{M}$, since $G$ has at most $r$ eigenvalues larger than $\lambda$, we have
    \begin{equation*}
        \E_{(u,v)\sim G}\bracks*{ \wt{M}_{uv} }
        \leq \sqrt{r \cdot \E_{u,v \sim [n]} \bracks*{\wt{M}_{uv}^2}} + \lambda
        \leq \sqrt{2r \delta} + \lambda \mper
        \qedhere
    \end{equation*}    
\end{proof}

%% file: tr-proof.tex
\section{Proof of \texorpdfstring{\Cref{thm:main}}{Theorem~\ref{thm:main}}: 3-coloring}

The following is the key lemma about $3$-colorings.
We emphasize that the analogous statement for $4$-coloring does not hold; see \Cref{rem:4-coloring-fails}.



\begin{lemma} \label{lem:mutual-info-lb}
    Let $\gamma, \eta \in (0,1/4)$.
    Let $X,Y$ be jointly distributed random variables on $[3] \cup \{\bot\}$.
    Suppose that $\Pr[X=\sigma]$, $\Pr[Y=\sigma] \leq \frac{1}{2}+\gamma$ for all $\sigma \in[3]$, and  moreover, $\Pr[X=\bot]$, $\Pr[Y=\bot] \leq \eta$.
    Then, 
    \begin{align*}
        \sum_{\sigma\in[3]} \Pr[X=\sigma] \cdot \Pr[Y=\sigma] \geq \frac{1}{4} - \frac{\eta}{2} - \gamma \mper
    \end{align*}
\end{lemma}
\begin{proof}
    We will use the rearrangement inequality, which states that for any sequences $a_1 \geq a_2 \geq \cdots \geq a_k$ and $b_1 \geq b_2 \geq \cdots \geq b_k$ and any permutation $\pi$ of $[k]$, $a_1 b_{\pi(1)} + a_2 b_{\pi(2)} + \cdots + a_k b_{\pi(k)} \geq a_1 b_k + a_2 b_{k-1} + \cdots + a_k b_1$.

    Consider the sequences $\{a_i\}_{i\in[3]} = \{\Pr[X=\sigma]\}_{\sigma\in[3]}$ and $\{b_i\}_{i\in [3]} = \{\Pr[Y=\sigma]\}_{\sigma\in[3]}$, sorted in descending order.
    Note that both sequences are non-negative and sum up to at least $1-\eta$.
    Then,
    \begin{align*}
        \sum_{\sigma\in[3]} \Pr[X=\sigma] \Pr[Y=\sigma] \geq a_1 b_3 + a_2 b_2 + a_3 b_1
        \geq a_2(b_2+b_3) + a_3 b_1
        = a_2 - b_1 (a_2 - a_3) \mcom
    \end{align*}
    where the first inequality uses the rearrangement inequality, and the second inequality follows from $a_1 \geq a_2$.
    Then, using $a_1,b_1 \leq \frac{1}{2}+\gamma$, $a_2+a_3 \geq 1-\eta-a_1 \geq \frac{1}{2}-\eta-\gamma$, and $0 \leq a_2-a_3 \leq a_2 \leq \frac{1}{2}$ (since $a_1 \geq a_2$ and $a_1+a_2 \leq 1$), we get a lower bound of
    \begin{align*}
        a_2 - \parens*{\frac{1}{2}+\gamma} (a_2 - a_3)
        = \frac{1}{2}(a_2+a_3) - \gamma (a_2-a_3)
        \geq \frac{1}{2}\parens*{\frac{1}{2}-\eta-\gamma} - \frac{\gamma}{2} = \frac{1}{4} - \frac{\eta}{2} - \gamma \mper
    \end{align*}
    This completes the proof.
\end{proof}

\begin{remark} \label{rem:4-coloring-fails}
    For 4-coloring, the analogous statement of \Cref{lem:mutual-info-lb} does not hold.
    For example, suppose $X$ is uniformly distributed over $\{1,2\}$ and $Y$ is uniformly distributed over $\{3,4\}$.
    Then, $\sum_{\sigma\in[4]} \Pr[X=\sigma] \Pr[Y=\sigma] = 0$.
    Indeed, $X$ and $Y$ are not concentrated on any single color, and they have zero correlation.
\end{remark}

We now prove our main result on coloring almost $3$-colorable graphs.

\begin{proof}[Proof of \Cref{thm:main}]
    Let $\lambda \coloneqq \eps/100$.
    We first solve the degree-$O(r/\lambda^2)$ sum-of-squares relaxation of the $\delta$-almost 3-coloring program on the input graph $G$, with alphabet $\Sigma = [3] \cup \{\bot\}$ and additional constraints that
    \begin{itemize}
        \item for any $(u,v)\in E(G)$, $\1(X_u = \sigma) \cdot \1(X_v=\sigma) = 0$ for all $\sigma\in [3]$,
        \item $\sum_{u\in[n]} \1(X_u = \bot) \leq \delta n$.
    \end{itemize}
    The SoS algorithm (\Cref{fact:sos-algorithm}) outputs a pseudo-distribution $\mu'$ that satisfies the above constraints (to $2^{-\poly(n)}$ precision), in time $n^{O(r/\eps^2)}$.

    After $O(r/\lambda^2)$ rounds of conditioning via \Cref{lem:ragh-tan}, we obtain a pseudo-distribution $\mu$ that satisfies the above constraints, and moreover, has small global correlation: $\E_{u,v\sim [n]}[I_{\mu}(X_u; X_v)] \leq \lambda^2/2r$.
    This procedure also takes time $n^{O(r/\eps^2)}$.

    Let $\gamma = 0.001$, and
    let $B \coloneqq \braces*{u\in[n]: \pPr_{\mu}[X_u = \bot] \geq \gamma}$.
    Since $\E_{u\sim [n]} \pPr_{\mu}[X_u = \bot] \leq \delta$ by the constraints, it follows that $|B| \leq O(\delta) n$ by Markov's inequality.
    Let
    \begin{align*}
        S_{\sigma} \coloneqq \braces*{ u\in [n]: \pPr_{\mu}[X_u = \sigma] \geq \frac{1}{2}+\gamma } \mcom
        \quad S \coloneqq S_1 \cup S_2 \cup S_3 \mcom
        \quad T \coloneqq [n] \setminus (S \cup B) \mper
    \end{align*}
    Note that $S_1, S_2,S_3$ must be disjoint.
    Moreover, by coloring $S_{\sigma}$ with color $\sigma$, we get a valid partial $3$-coloring of vertices in $S$.
    This is because for any edge $(u,v)\in E(G)$, if both $\pPr_{\mu}[X_u=\sigma]$ and $\pPr_{\mu}[X_v = \sigma] \geq \frac{1}{2} + \gamma$ for some $\sigma\in[3]$, then by a union bound, we must have $\pPr_{\mu}[X_u = X_v = \sigma] \geq 2\gamma > 0$, which is a contradiction.
    Thus, it suffices to show a lower bound on $|S|$.

    \parhead{Prove that $|S| \geq (\frac{1}{2}-\eps-O(\delta)) n$.}
    The input graph $G$ has at most $r$ eigenvalues larger than $\lambda$.
    Since $\E_{u,v\sim [n]}[I_{\mu}(X_u; X_v)] \leq \lambda^2/2r$, by \Cref{lem:local-correlation}, we have an upper bound on the ``local correlation'':
    \begin{align*}
        \E_{(u,v)\sim G} \sum_{\sigma\in [3]} 
        \parens*{ \pPr_{\mu}[X_u = X_v = \sigma] - \pPr_{\mu}[X_u = \sigma] \pPr_{\mu}[X_v = \sigma] }^2
        \leq \sqrt{2r \cdot \frac{\lambda^2}{2r}} + \lambda
        \leq \frac{\eps}{50} \mper
        \numberthis \label{eq:local-corr-bound}
    \end{align*}
    For any edge $(u,v) \in E(G)$, we must have $\pPr_{\mu}[X_u = X_v = \sigma] = 0$, so the left-hand side of \Cref{eq:local-corr-bound} is exactly
    $\E_{(u,v)\sim G} \sum_{\sigma\in[3]} \pPr_{\mu}[X_u = \sigma]^2 \cdot \pPr_{\mu}[X_v = \sigma]^2$.

    Let us denote $M_{uv} \coloneqq \sum_{\sigma\in[3]} \pPr_{\mu}[X_u = \sigma]^2 \cdot \pPr_{\mu}[X_v = \sigma]^2$.
    For any $(u,v)\in E(G)$ and $u,v \in T$, we have $\pPr_{\mu}[X_u = X_v = \sigma] = 0$ and $\pPr_{\mu}[X_u=\sigma]$, $\pPr_{\mu}[X_v=\sigma] < \frac{1}{2}+\gamma$ for all $\sigma\in [3]$, and $\pPr_{\mu}[X_u=\bot]$, $\pPr_{\mu}[X_v=\bot] \leq \gamma$.
    Thus, applying \Cref{lem:mutual-info-lb}, we get
    \begin{align*}
        M_{uv} \geq \frac{1}{3} \parens*{\sum_{\sigma\in[3]} \pPr_{\mu}[X_u = \sigma] \cdot \pPr_{\mu}[X_v = \sigma]}^2
        \geq \frac{1}{3} \parens*{\frac{1}{4} - \frac{\gamma}{2} -\gamma}^2
        > \frac{1}{50} \mper
    \end{align*}
    Thus, all edges $(u,v)$ fully contained in $T$ have $M_{uv} > 1/50$.
    Let $d$ be the degree of the graph $G$.
    Then, counting the edges between $S \cup B$ and $T$,
    \begin{align*}
        e(S\cup B, T) &= d|T| - 2e(T) = d|S\cup B| - 2e(S \cup B) \\
        \implies 2e(T) &\geq (|T|-|S\cup B|) d = (n-2|S\cup B|) d \mper
    \end{align*}
    This implies that
    \begin{align*}
        \E_{(u,v)\sim G}[M_{uv}]
        \geq \frac{1}{50} \cdot \frac{2e(T)}{nd} \geq \frac{1}{50} \parens*{1 - \frac{2(|S|+|B|)}{n}} \mper
    \end{align*}
    On the other hand, \Cref{eq:local-corr-bound} states that $\E_{(u,v)\sim G}[M_{uv}] \leq \eps/50$.
    This proves that $|S| \geq (\frac{1}{2}-\eps-O(\delta)) n$.
\end{proof}

%% file: tr-is.tex
\section{Proof of \texorpdfstring{\Cref{thm:main-is}}{Theorem~\ref{thm:main-is}}: large independent set}

In this section, we prove our result on finding large independent sets in graphs containing $(\frac{1}{2}-\delta)n$-sized independent sets.

\begin{proof}[Proof of \Cref{thm:main-is}]
    Let $\lambda \coloneqq \eps^5/100$.
    We first solve the degree-$O(r/\lambda^2)$ sum-of-squares relaxation of the independent set program on the input graph $G$, with alphabet $\{0,1\}$ and additional constraints that
    \begin{itemize}
        \item for any $(u,v)\in E(G)$, $\1(X_u = 1) \cdot \1(X_v = 1) = 0$,
        \item $\sum_{u\in[n]} \1(X_u = 1) \geq (\frac{1}{2}-\delta) n$.
    \end{itemize}
    The SoS algorithm (\Cref{fact:sos-algorithm}) outputs a pseudo-distribution $\mu'$ that satisfies the above constraints (to $2^{-\poly(n)}$ precision), in time $n^{O(r/\lambda^2)}$.

    After $O(r/\lambda^2)$ rounds of conditioning via \Cref{lem:ragh-tan}, we obtain a pseudo-distribution $\mu$ that satisfies the above constraints, and moreover, has small global correlation: $\E_{u,v\sim [n]}[I_{\mu}(X_u; X_v)] \leq \lambda^2/2r$.
    This procedure also takes time $n^{O(r/\lambda^2)}$.

    Let $\gamma = \eps/100$ (which is at least $1/\poly(n)$ otherwise the theorem statement is trivial).
    Let
    \begin{align*}
        S \coloneqq \braces*{u\in [n]: 
        \pPr_{\mu}[X_u = 1] \geq \frac{1}{2}+\gamma } \mcom
        \qquad
        A \coloneqq \braces*{u\in [n]: 
        \pPr_{\mu}[X_u = 1] \leq \eps/2 } \mper
    \end{align*}
    Note that $S$ forms an independent set.
    Thus, it suffices to show a lower bound on $|S|$.

    \parhead{Prove that $|S| \geq (\frac{1}{2}-\eps-2\delta)n$.}
    Recall that we have $\sum_{u\in [n]} \pPr_{\mu}[X_u=1] \geq (\frac{1}{2}-\delta)n$, and we  have the upper bound
    \begin{align*}
        \sum_{u\in [n]} \pPr_{\mu}[X_u=1]
        &\leq \frac{\eps}{2} |A| + |S| + \parens*{\frac{1}{2}+\gamma} (n-|S|-|A|)  \\
        &\leq \parens*{\frac{1}{2}+\gamma} n + \frac{1}{2} |S| - \parens*{\frac{1}{2}-\frac{\eps}{2}} |A| \mper
    \end{align*}
    Suppose $|A| \geq (\frac{1}{2}-\frac{\eps}{3})n$, then we have
    \begin{align*}
        |S| \geq (1-\eps)|A| - 2(\delta+\gamma)n
        \geq \parens*{\frac{1}{2}-\eps - 2\delta} n \mper
    \end{align*}
    Next, we show that $|A| < (\frac{1}{2}-\frac{\eps}{3})n$ would lead to a contradiction.

    The input graph $G$ has at most $r$ eigenvalues larger than $\lambda$.
    Since $\E_{u,v\sim [n]}[I_{\mu}(X_u; X_v)] \leq \lambda^2/2r$,
    by \Cref{lem:local-correlation}, we have an upper bound on the ``local correlation'':
    \begin{align*}
        \E_{(u,v)\sim G} \sum_{\sigma\in \{0,1\}} 
        \parens*{ \pPr_{\mu}[X_u = X_v = \sigma] - \pPr_{\mu}[X_u = \sigma] \pPr_{\mu}[X_v = \sigma] }^2
        \leq \sqrt{2r \cdot \frac{\lambda^2}{2r}} + \lambda
        \leq \frac{\eps^5}{50} \mper
        \numberthis \label{eq:local-corr-bound-is}
    \end{align*}
    For any edge $(u,v) \in E(G)$, we must have $\pPr_{\mu}[X_u = X_v = 1] = 0$, so the left-hand side of \Cref{eq:local-corr-bound-is} is exactly
    $\E_{(u,v)\sim G} \pPr_{\mu}[X_u = 1]^2 \cdot \pPr_{\mu}[X_v = 1]^2$.
    Let us denote $M_{uv} \coloneqq \pPr_{\mu}[X_u = 1]^2 \cdot \pPr_{\mu}[X_v = 1]^2$.    
    For any edge $(u,v)\in E(G)$ and $u,v \notin A$, we have $M_{uv} \geq (\eps/2)^4$.
    
    Let $T = [n] \setminus A$, which has size at least $(\frac{1}{2}+\frac{\eps}{3})n$.
    Let $d$ be the degree of the graph $G$.
    Then, counting the edges between $A$ and $T$,
    \begin{align*}
        e(A, T) &= d|T| - 2e(T) = d|A| - 2e(A)
        \leq d|A| \\
        \implies 2e(T) &\geq d(|T|-|A|) > \frac{2\eps}{3}nd \mper
    \end{align*}
    Then,
    \begin{align*}
        \E_{(u,v)\sim G}[M_{uv}] \geq (\eps/2)^4 \cdot \frac{2e(T)}{nd} > \frac{\eps^5}{24} \mper
    \end{align*}
    This contradicts \Cref{eq:local-corr-bound-is}, finishing the proof.
\end{proof}